\documentclass[12pt]{article}
\usepackage{amscd,amsfonts,amsmath,amssymb,delarray,graphics,latexsym,lipsum,rotating,subfigure}
\usepackage{geometry}
\usepackage[usenames]{color}
\newtheorem{thm}{Theorem}[section]

\newtheorem{lem}{Lemma}[section]

\newenvironment{proof}[1][Proof]{\textbf{#1:} }{\ \rule{0.5em}{0.5em}}

\newcommand{\EP}{\mathbb{EP}}
\numberwithin{equation}{section}

\title{On Fair Reinsurance Premiums; Capital Injections in a Perturbed Risk Model}

\author{Zied Ben Salah\footnote{Corresponding author: Department of Mathematics and Actuarial Sciences. American University in Cairo, P.O. Box 74, New Cairo 11835, Egypt. Email: zied.bensalah@aucegypt.edu. 
This author gratefully acknowledges the financial support of a B3 Postdoctoral Fellowship from the Fonds de recherche du Qu\'ebec - Nature et technologies (FRQNT).} 
\quad and \quad
\setcounter{footnote}{6}
Jos\'{e} Garrido\footnote{Department of Mathematics and Statistics. Concordia University, 1455 de Maisonneuve Blvd W, Montreal, Quebec H3G 1M8, Canada. Email: jose.garrido@concordia.ca. This author gratefully acknowledges the financial support of the Natural Sciences and Engineering Research Council (NSERC) of Canada grant 36860-2017.} 
}

\date{}
 
\usepackage{graphicx}
\graphicspath{%
    {converted_graphics/}
    {/}
}
\begin{document}

\maketitle

\begin{abstract}

We consider a risk model where deficits after ruin are covered by a new type of reinsurance contract that provides capital injections. To allow the insurance company's survival after ruin, the reinsurer injects capital only at ruin times caused by jumps larger than a chosen retention level. Otherwise capital must be raised from the shareholders for small deficits. The problem here is to determine adequate reinsurance premiums. It seems fair to base the net reinsurance premium on the discounted expected value of any future capital injections. 
Inspired by the results of Huzak \emph{et al.}~(2004) and Ben Salah (2014) on successive ruin events, we show that an explicit formula for these reinsurance premiums exists in a setting where aggregate claims are modeled by a subordinator and a Brownian perturbation. Here ruin events are due either to Brownian oscillations or jumps and reinsurance capital injections only apply in the latter case. The results are illustrated explicitly for two specific risk models and in some numerical examples.
\text{}
\end{abstract}

\section{Introduction}

Reinsurance contracts between a direct insurer and a reinsurer are used to transfer part of the risks assumed by the insurer. The problematic risks are those carrying either the possible occurrence of very large individual losses, the possible accumulation of many losses from non-independent risks, or those from other occurrences that could prevent insurers from fulfilling their solvency requirements. So traditionally reinsurance has been an integral part of insurance risk management strategies (see Centeno and Sim\~oes, 2009 for a survey of the different types of reinsurance and recent optimal reinsurance results). However, over time, global financial markets have developed additional or alternative risk transfer mechanisms, such as swaps, catastrophe bonds or other derivative products, that have helped insurers reduce their risk mitigation costs.

In this spirit of designing possibly cheaper risk transfer agreements we consider here a new type of reinsurance contract that would provide capital injections only in extreme, worse scenario cases. It differs from excess--of--loss (XL) agreements, or even catastrophe XL (Cat XL), in that it is neither a per--risk nor a per--event reinsurance contract, but rather one based on the insurer's financial position. Here ruin will serve as a simplifying proxy for the insurer's financial health. Reinsurance capital injections, after ruin, would allow the insurance company to continue operate until the next ruin. Again to simplify the analysis we adopt an on-going concern basis and set an infinite horizon for the reinsurance treaty, which can allow repeated ruin events. The reinsurance agreement then calls for a capital injection after these successive ruin events, keeping the insurer afloat in perpetuity. We call this new type of agreement {\it reinsurance by capital injections} (RCI).       

Here our jump--diffusion surplus process can generate two types of ruin events, hence different 
covers are assumed with distinct sources of capital. Surplus fluctuations due to jumps 
are assumed to represent larger claim costs from events unfavorable to the insurer; a ruin 
caused by such jumps will trigger a capital injection from the extreme-loss 
reinsurance contract, at ruin time, if the capital injection is larger than a certain threshold 
(retention limit). By contrast, Brownian oscillations represent comparatively smaller surplus 
fluctuations; so ruin caused by oscillations should be easier to cover with capital raised 
directly from the stockholders. Hence the reinsurer does not provide capital injections in cases 
when (1) ruin is from an oscillation, or (2) when it is from a jump producing a capital injection 
smaller than the threshold. As explained in the paper, even if stockholders may need cover these 
2 types of ruin costs at first, they may ultimately get reimbursed by the reinsurer, at a subsequent ruin 
time due to a jump, if the latter is deep enough to meet the threshold. 

Two recent developments in the literature make the analysis of the RCI contracts now possible, in the sense of getting tractable formulas for net premiums that would be fair to both parties for such agreements. The first one is the development of actuarial and financial models for capital injections (see for instance Einsenberg and Schmidli, 2011, or more recently Avram and Loke, 2018, and the references therein) and the other is the derivation of tractable formulas for the expected present value of future capital injections in a quite general class of risk models (see Huzak \emph{et al.}, 2004, and Ben Salah, 2014). The application presented here builds on this recent theory to develop fair lump sum net premiums for two types of RCI contracts, both over an infinite horizon. In practice our net premiums would have to be allocated to finite policy terms (e.g.~a year) and loaded appropriately to define gross (market) premiums. In this first study we focus on the definition of the RCI contracts and the derivation of the premium formulas so that both, insurance and reinsurance companies, can compare the cost of RCI contracts to their alternative risk mitigation strategies/products. Future work would then need to address the issue of optimizing the insurance firm value by weighing these premiums in relation to other concurrent capital injections from shareholders.  

To sum up, the paper is organized as follows: the general risk model used here is defined in Section \ref{riskEDPF}. Then Section \ref{preliminary result} covers the preliminary technical results needed to derive the expected present value of future capital injections. Section \ref{EDVCI def} gives the main result, with the derivation of fair premiums for reinsurance based on capital injections in the general risk model defined in Section \ref{riskEDPF}. These are illustrated in detail for two classical risk processes in Section \ref{classical}, which gives also numerical illustrations. The article concludes with some general remarks. 

\section{Risk model}\label{riskEDPF}

We consider a general insurance surplus model that extends the standard Cram\'er--Lundberg theory to allow for jumps and diffusion type fluctuations. Here 
	\begin{equation}
R_t := x - Y_t \;, \qquad t \geqslant 0 \,,
	\label{eq:risk}
	\end{equation} 
where $x\geqslant 0\/$ is the initial surplus and the risk process 
$Y\/$, a spectrally positive L\'evy process defined on a filtered probability space $(\Omega, \mathcal{F},(\mathcal{F}_t)_{t\geqslant 0}, \mathbb{P})\/$, is given by
	\begin{equation}
Y_t := -c\;t + S_t +\sigma B_t \;, \qquad t \geqslant 0 \,,
	\label{riskmodel3} 
	\end{equation} 
where $S=(S_t)_{ t\geq 0}\/$ is a subordinator (i.e.~a L\'evy process of bounded variation and non--decreasing paths) without a drift ($S_0=Y_0=0\/$) and $B\/$ is a standard Brownian motion independent of $S\/$. Let $\nu\/$ be the L\'evy measure of $S\/$; that is, $\nu\/$ is a $\sigma$--finite measure on $(0, \infty)\/$ satisfying $ \int_{(0,\infty)}(1\wedge y )\nu(dy) < \infty\/$. 
In this case the Laplace exponent of $S\/$ is defined by 
	$$\psi_S(s) = \int_{(0,\infty)} (e^{s\,y} -1 ) \,\nu(dy)\,,$$
where $\mathbb{E}[e^{s S_t}] = e^{ t\,\psi_S(s)}\/$. 


Note that the risk process in \eqref{eq:risk} is similar in spirit to the original perturbed surplus process introduced in Dufresne and Gerber (1991).
The constant $x > 0\/$ represents the initial surplus, while the process $Y\/$ models the cash outflow of the primary insurer and the subordinator $S\/$ represents aggregate claims. That is why $S\/$ needs to be an increasing process, with the jumps representing the claim amounts paid out. The Brownian motion $B\/$ accounts for any small fluctuations affecting other components of the risk process dynamics, such as the claim arrivals, premium income or investment returns. 

Here $c\,t\/$ represents aggregate premium inflow over the interval of time $[0,t]\/$.
The premium rate $c\/$ is assumed to satisfy the net profit condition, 
more precisely $\mathbb{E}[S_1]<c\/$, which means that
	\begin{equation}
\int_{(0,\infty)} y\,\nu(dy) < c \,.
	\label{netprofit}
	\end{equation}
Condition \eqref{netprofit} implies that the process $Y\/$ has a negative drift, in order to avoid the possibility that $R\/$ becomes negative almost surely. This condition is often expressed in terms of a safety loading applied to the net premium. 
For instance, note that we can recover the classical Cram\'er--Lundberg model if $\sigma=0\/$ and $c:=(1+\theta)\,\mathbb{E}[S_1]\/$, for $S\/$ a compound Poisson process modeling aggregate claims. 

We do not use the concept of safety loading in this paper, in order to simplify the notation, but we stress the fact that this concept is implicitly considered within the drift of $Y\/$ when we impose condition \eqref{netprofit}. The classical compound Poisson model is a special case of this framework where $\nu(dy)=\lambda\,K(dy)\/$, with $\lambda\/$ being the Poisson arrival rate and $K\/$ a diffuse claim distribution. We refer to Asmussen and Albrecher (2010) for an account on the classical risk model, and to Dufresne and Gerber (1991), Dufresne, Gerber and Shiu (1991), Furrer and Schmidli (1994), Yang and Zhang (2001),  Biffis and Morales (2010) and Ben Salah (2014) for the original and different generalizations or studies of the model in \eqref{riskmodel3}.

Now, one of the main objectives of this paper is to obtain an expression for the reinsurance premium for the risk model in \eqref{riskmodel3}. 
First we need to define quantities and notation associated with the ruin time, as well as the sequence of times of successive deficits due to a claim of the surplus process \eqref{riskmodel3} after ruin. 
Let  $\tau_x\/$ be the \emph{ruin time} representing the first passage time of $R_t\/$ below zero when $R_0=x\/$, i.e. 
	\begin{equation}
\tau_x := \inf\{ t > 0 \, : \, Y_t > x \}\,,
	\label{def:ruin}
	\end{equation}
where we set $\tau_x=+\infty\/$ if $R_t\geq0\/$, for all $t\geq 0\/$. 
We define the first new record time of the running supremum
	\begin{equation}
\tau := \inf \{ t>0 \, : \, Y_t>\overline{Y}_{t^-}\} \,,
	\label{T}
	\end{equation}
and the sequence of times corresponding to new records of $Y\/$ (that is $\overline{Y}_t := \sup\{Y_s \, : \, t \ge s\}\/$) due to a jump of $S\/$ after the ruin time $\tau_x\/$. More precisely, let
	\begin{equation}
\tau^{(1)} := \tau_x\,,
	\label{T1}
	\end{equation}
and, assuming that $\{ \tau^{(n)} < \infty \}$, then by induction on $n \ge 1\/$:
	\begin{align}\label{T2}
\tau^{(n+1)}:=\inf \{ t>\tau^{(n)} \; : \; Y_t>\overline{Y}_{t^-}\}\;,
	\end{align}
(note that by this definition $\tau^{(1)}\/$ differs from the consecutive new record
times $(\tau^{(n)})_{n>1}\/$; the former includes ruin events caused by jumps and Brownian
oscillations, while the latter include subsequent records only due to jumps).

Recall from Theorem 4.1 of Huzak \emph{et al.}~(2004) that the sequence $(\tau^{(n)})_{n>1}\/$ is discrete, and, in particular, neither time $0\/$ nor any other time is an accumulation point of these $\tau^{(n)}$'s. More precisely, $\tau > 0\/$ a.s.~and $\tau^{(n)}<\tau^{(n+1)}\/$ a.s.~if $\{ \tau^{(n)} <\infty \}\/$. As a consequence, we can order the sequence $(\tau^{(n)})_{n\ge1}\/$ of times when a new supremum is reached by a jump of a subordinator as $ 0<\tau^{(1)}<\tau^{(2)}<\cdots \quad$ a.s.; see Figure \ref{fig:Image1}.

Finally, consider the random number 
	\begin{equation}
N: = \max\{ n: \tau^{(n)} < \infty \}\,,
	\label{N}
	\end{equation}
which represents the number of new records reached by a claim of the surplus process in \eqref{riskmodel3}. 

\begin{figure}[tbp] 
  \centering
  \includegraphics[width=3.52in,height=2.32in,keepaspectratio]{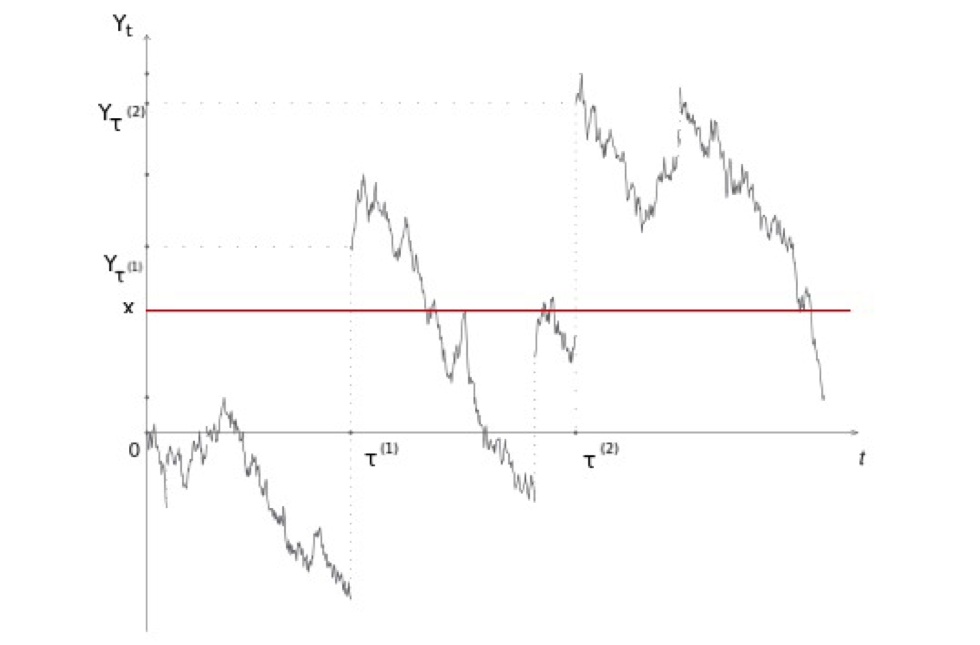}
  \caption{Sample path of $Y_t=-ct+S_t+\sigma B_t$ in \eqref{riskmodel3}}
  \label{fig:Image1}
\end{figure}

Before developing fair premiums for these new reinsurance by capital injections contracts that we define here, the next section first presents the theory available for the spectrally negative L\'evy risk model defined in \eqref{riskmodel3}; see \cite{Doney Kyprianou}, \cite{Kyp} and \cite{Avram} for more details.  

\section{Preliminary results}\label{preliminary result}

This section reviews some notions and results needed in the rest of the paper. Let $X=(X_t)_{t\geqslant 0}\/$ be a spectrally negative L\'evy process defined by 
	$$X_t = -Y_t = c\,t -S_t-\sigma B_t\,.$$
Since $X\/$ has no positive jumps, the expectation $\mathbb{E}[e^{s X_t}]\/$ exists for all $s\geqslant 0\/$ and is given by $\mathbb{E}[e^{s X_t}]= e^{ t\,\psi(s)}\/$, where $\psi(s)\/$ is of the form 
	\begin{equation}
\psi(s) = c\,s + \frac{1}{2} \sigma^2 s^2 +\int_{0}^{\infty} (e^{-x\,s} -1 ) \,\nu(dx)\,.
	\label{khinchine}
	\end{equation}
Here $c\in \mathbb{R}\/$, $\sigma>0\/$ and $\nu\/$ is the L\' evy measure associated with the process $Y\/$ (for a thorough account of L\'evy processes see \cite{Bertoin, Kyp}). 

Consider $\Phi$, the right inverse of $\psi\/$, defined on $[0,\infty)\/$ by
	\begin{equation}
\Phi(q) := \sup\{s \geqslant 0 \,:\, \psi(s)=q\}\,.
	\label{inverse}
	\end{equation}
Note that since $X\/$ is a spectrally negative L\'evy process $X\/$, we have that $\Phi(q)> 0\/$ for $q > 0\/$ (see \cite{Kyp}).

It is well--known that, for every $q\geqslant 0\/$, there exists a function $W^{(q)} \,:\, \mathbb{R} \longrightarrow [0, \infty)\/$ such that $W^{(q)}(y)=0\/$, for all $y < 0\/$ satisfying
	\begin{equation}
\int_0^{\infty} e^{-\lambda y} \,W^{(q)}(y)\,dy = \frac{1}{\psi(\lambda)-q}, \qquad \lambda >\Phi(q)\,.
	\label{scale1}
	\end{equation}
These are the so--called \emph{$q$-scale functions} $\{W^{(q)},\;q\geqslant 0\}\/$ of the process $X\/$ (see \cite{Kyp}), a key notion in the analysis of passage times for spectrally negative L\'evy processes. Note that for $q=0\/$, equation \eqref{scale1} defines the so--called scale function and we simply write $W := W^{(0)}\/$. 

The following theorem plays a key role here (we refer to \cite{Bensalah1} for a thorough discussion and the proof). It gives an expression for a general form of the expected discounted penalty function (EDPF), $\EP(F,q,x)\/$, defined in \cite{Bensalah1} as
	\begin{equation}
\EP(F,q,x) = \mathbb{E}\Big[\sum_{n=1}^N e^{-q\,\tau^{(n)}} F_n( Y_{\tau^{(n-1)}},Y_{\tau^{(n)}} ) \, : \, \tau_x < \infty  \Big]\,,
	\label{def}
	\end{equation}
where $ F=(F_n)_{n \geq 0}\/$ is a sequence of given non--negative measurable functions from $\mathbb{R}_{+} \times \mathbb{R}_{+}$ to $\mathbb{R}\/$, and where $x,\,q \geq 0\/$ are also given. This generalizes the prior results on the EDPF in \cite{GS3} and \cite{Garrido}.

Let $H \ast G(\cdot)\/$  denote the convolution of $H(\cdot)\/$ with $G(\cdot)\/$ defined by 
	$$\int_{A} f(u)\,H \ast G(du) = \int_{\{y+v\in A\}} f(y+v)\, H(dy)\,H(dv)\,,$$
for any Borel set $A\/$ of $\mathbb{R}\times \mathbb{R}\/$ and nonnegative, bounded Borel function $f\/$. As usual $f^{\ast n}\/$, for $n\geq 1\/$, denotes the $n$--fold convolution of $f\/$ with itself and $f^{\ast 0}\/$ is the distribution function corresponding to the Dirac measure at zero. For more details about this result, we refer to \cite{Bensalah1}.

	\begin{thm}
Consider the risk model in \eqref{riskmodel3}. For $x,\,q\geq 0\/$, the extended EDPF is given by
	\begin{eqnarray}
\EP(F,q,x) &=& \phi(w,q,x) + \sum_{n=0}^\infty \int_{(x,\infty)} \int_{(0,\infty)} e^{\Phi(q)\,(u+v)} F_{n+2}( v,u+v ) \nonumber \\
	&&\qquad \times H\ast G(du)\, H^{\ast n}\ast G^{\ast n}\ast T_x(dv)\,, \label{SEDPF}
	\end{eqnarray}
where the penalty $w\/$ is a measurable Borel--function, $F_1(z,y)=w(x-z,y-x)\/$, for  $y\geq x \geq z\/$  and $\phi(w,q,x)\/$ is the classical EDPF (Gerber--Shiu) function defined by
	\begin{equation}
\phi(w,q,x) = \mathbb{E}\Big[e^{-q\,\tau_x}\,w( x-Y_{\tau_x^-},Y_{\tau_x}-x ) \,;\, \tau_x < \infty  \Big]\,,
	\label{edpfs1}
	\end{equation}
	\label{characterization}
(see \cite{GS}--\cite{GS4} for details on the classical Gerber--Shiu function). 
	\end{thm}
We assume in Theorem \ref{characterization} that $\tau^{(0)}=\tau_x^-\/$ and 
	\begin{equation}
F_1(\cdot,x)=0 \quad\text{and}\quad F_n(y,y)=0\,, \qquad\text{for }y\in(x,\infty)\,,\;n>1\,.
	\label{cdt1}
	\end{equation}
Note that condition \eqref{cdt1} is used to exclude the events $\{Y_{\tau_x}=x\}\/$ and $\{Y_{\tau^{(n-1)}}=Y_{\tau^{(n)}}\}\/$. Also note that $\EP(F,q,x)\/$ is an extension of the classical EDPF $\phi(w,x,q)$ in \eqref{edpfs1}. In particular, it reduces to it if $F_1(u,v)=w(x-u, v-x)\/$ and $F_n=0\/$ for $n\geq 2\/$.

Here the density $H(\cdot)\/$ is given by
	$$H(du) = \frac{e^{-\Phi(q)\/u}}{c+\Phi(q)\sigma^2} \int_0^\infty e^{-\Phi(q)\,y}\,\nu(du+y)\,dy\,,\qquad u>0\,,$$
and $G(\cdot)\/$ is an exponential distribution function with parameter $2\widetilde{c}/\sigma^2\/$, where $\widetilde{c}=c+\sigma^2\,\Phi(q)\/$.

Then $T_x(\cdot)\/$ is the distribution of the overshoot at $\tau_x\/$ under the probability measure $ \widetilde{\mathbb{P}}\/$, with density process defined by
	\begin{equation}
\frac{d\widetilde{\mathbb{P}}}{d\mathbb{P}}\Big|_{\mathcal{F}_t} = e^{-\Phi(q)\,Y_t-q\,t}\,, 
	\label{change measure}
	\end{equation}
where $\Phi(q)\/$ is the right inverse of $\psi\/$ defined in \eqref{inverse}. The overshoot distribution $T_x(\cdot)\/$ is thus given by
	\begin{equation}
T_x(du) = \int_0^\infty \int_0^v e^{-\Phi(q)\,u}\,\nu(du-x+v)\,\big[W'^{(q)}(x-y)-\Phi(q)\,W^{(q)}(x-y)\big]\,dy\,dv\,.
	\label{overshoot}
	\end{equation}

\section{Fair premiums for reinsurance by capital injections (RCI)}\label{EDVCI def}

As in a standard insurance contract, here the reinsurer charges a premium to the insurer, that is  equal or larger than the expected value of the ceded risk. There is then a trade-off between the risk retained by the insurer and the premium paid to the reinsurer. Given a fixed retention, determining the optimal premium is an important issue for the reinsurer in such a context. The main objective of this paper is to derive an expression for a fair premium for this new reinsurance by capital injections (RCI) contract, for the risk model in \eqref{riskmodel3}. The premium derived here is in net terms over an infinite contract term. For annual or other short--term reinsurance contracts, our premium would have to be allocated to each annual/other interval. This is however beyond the scope of this first study, as is the optimization
of the choice of reinsurance retention/premium for these RCI treaties. 

We consider two types of RCI contracts, {\it proportional} RCI and {\it extreme--loss} RCI, which can be combined to design contracts that would be more appropriate in practice. For proportional RCI, the ``proportion" of risk ceded to the reinsurer for a claim of size $C\/$ is $a\,C\/$ , where $a < 1\/$. While in the extreme--loss case the ceded amount of risk to the reinsurer is $C \;\mathbb I_{\{ C\geq m\}}\/$ , where $m \geq 0\/$. That is, the ceded amount of risk is the total deficit $C\/$ if $C\/$ exceeds a certain level of retention $m \geq 0\/$, and $0\/$ otherwise. The reason to set $a < 1\/$ in the proportional RCI contract is to avoid moral hazard, a problem not present with the proposed extreme--loss RCI design. Note, however that from a purely mathematical point of view, the proportional RCI premium formulas below are also valid for $a \ge 1\/$, so that the capital injections provide the insurer sufficient funds to recover from ruin and restart from a solvent position, without a need to raise additional capital from stockholders. 

Based on the above preliminary results, an explicit form for fair RCI net premiums, defined as the discounted value of future capital injections, is: 
	\begin{equation}
\Pi(q,x,r(\cdot)) = \mathbb{E}\Big[ \sum_{n=0}^N  e^{-q\,\tau^{(n+1)}} r(C_n) \Big]\,,
	\label{EDVCI}
	\end{equation}
where $r(C_n)\/$ is given either by $r(C_n)=a\,C_n\/$ (proportional case) or $r(C_n)=C_{n}\, \mathbb I_{\{ C_{n}\geq m\}}\/$ (extreme--loss). Note that here $C_n\/$ denotes the size of the  $n$-th reinsurance claim after ruin (the $n$-th capital injection), that is  
	\begin{eqnarray}
C_n &=& R_{\tau^{(n)}} - R_{\tau^{(n+1)}} \nonumber\\
	&=& Y_{\tau^{(n+1)}} - Y_{\tau^{(n)}}, \qquad \text{ for } n\geq 1\,,
	\end{eqnarray}
and $C_0=Y_{\tau^{(1)}}-x\/$, for $n=0\/$, where $(\tau^{(n)})_{n\geq1}\/$ is the sequence of insurer's claim times corresponding to new records of $Y\/$, defined in Section \ref{riskEDPF}, in Equations \eqref{T1} and \eqref{T2}.

Also note that the extreme--loss capital injections occur at times $\tau^{(n)}\/$, of new records set by jumps. No capital injections are received if the surplus creeps below the threshold barrier only due to Brownian oscillations, without jumps. However, the cost of the Brownian oscillations between 2 record jumps is always included in the capital injection at the next record jump. The interpretation here is that in our model Brownian oscillations represent (smaller) 
capital requirements that are less likely to amount to a deficit creeping over the solvency 
threshold $m\/$, and hence such capital can be provided more easily by stockholders. By contrast,  
the subordinator jumps here represent larger (less predictable) losses that should cross 
the threshold more frequently and/or more deeply and for which the insurer needs the 
reinsurance capital injections. Clearly such a reinsurance scheme is best suited for 
companies with observed surplus experience where deficits caused by jumps dominate those 
caused by oscillations.

The following theorem gives expressions for the RCI premiums in both, the proportional and extreme--loss cases, in terms of $q$--scale functions, where $q\/$ is the present value discounting rate and the L\'{e}vy measure. This is the  main contribution and it is based on the result of Theorem \ref{characterization}.
	\begin{thm} \label{CI}
Consider the risk model introduced in \eqref{riskmodel3}: 
\\
1. The extreme--loss RCI reinsurance premium for $r(C_n)=C_n\, \mathbb I_{\{ C_n\geq m\}}\/$ in \eqref{EDVCI} is given by
	\begin{equation}
\Pi_1(q,x,m) = \varphi(q,x,m) + \delta(q, \sigma,m)\,\kappa(q,x)\,,
	\label{V2}
	\end{equation}
where
\begin{eqnarray*}
\delta(q,\sigma,m)&=&\frac{ c\; \big[2q +\Phi(q)(\rho+ m-1)(2c + \Phi(q)\sigma^2) \big]}{q\Phi(q)(2c+\Phi(q)\sigma^2)}\;e^{-\frac{2\;c + \phi(q)\sigma^2}{\sigma^2}m}  \\
&& \quad +\frac{2c+\Phi(q)\sigma^2 }{q\sigma^2} \int_0^m e^{-\frac{2c+\Phi(q)\sigma^2}{\sigma^2}v} \; \Big[ \int_0^{\infty}[1+(m\Phi(q)-1)e^{-\Phi(q)u}] \\ 
&&\qquad \times \nu(u+m-v,\;\infty)\;du \Big]\,dv\,.
\end{eqnarray*}
\\
2. The proportional RCI reinsurance premium for $r(C_n)= a\,C_n\/$ in \eqref{EDVCI} is given by
	\begin{equation}
\Pi_2(q,x,a) = a\,\varphi(q,x,0) + a\,\delta(q, \sigma, 0)\,\kappa(q,x)\,. 
	\label{V}
	\end{equation}
The functions $\kappa(q,x)\/$ and $\varphi(q,x,m)\/$ are given explicitly in terms of the $q$--scale function and  the L\'{e}vy measure as: 
	\begin{equation}
\varphi(q,x,m) = f\ast h_m(x)
	\label{varphi}
	\end{equation}
and 
	\begin{equation}
\kappa(q,x)= f\ast t(x)\,,
	\label{kappa}
	\end{equation}
where
	$$f(x)= W'^{(q)}(x)-\Phi(q)\,W^{(q)}(x)\,,$$
	\begin{equation}
h_m(x) = e^{\Phi(q)\,x}\,\int_x^\infty e^{-\Phi(q)\,v} \int_{(v,\infty)}(u-v)\,\nu(du+m)\, dv
	\label{h}
	\end{equation}
and
	\begin{equation}
t(x) = e^{\Phi(q)\,x} \int_x^\infty e^{-\Phi(q)\,v}\,\nu(v,\infty)\,dv\,.
	\label{t2}
	\end{equation}
	\end{thm}
\begin{proof}
1. Consider the sequence of functions $F_1(\cdot,y)=(y-x)\,\mathbb I_{\{y-x\geq m\}}\/$ for $y>x\/$ and $F_n(v,u)=(u-v)\,\mathbb I_{\{u-v\geq m\}}\/$, for $u \geq v > x\/$ and $n\geq 2\/$, then the extended EDPF associated with $F\/$ and $q\/$, $\EP(F,q,x)\/$, in Theorem \ref{characterization}, is equal to the extreme--loss RCI reinsurance premium $\Pi_1(q,x,m)\/$. In fact, using Theorem \ref{characterization}, we can  derive \eqref{V2}:
	\begin{eqnarray}\label{id1}
\Pi_{1}(q,x,m)&=&\mathbb{E}\Big[ e^{-q\tau_x}\,(Y_{\tau_x}-x )\,\mathbb I_{\{Y_{\tau_x}-x\geq m\}};  \tau_x < \infty  \Big]+\sum_{n=0}^{\infty} \int_{(x,\infty)}\int_{(0,\infty)} \nonumber\\
	&& e^{\Phi(q)(u+v)} \,u \,\mathbb I_{\{u\geq m\}}\, (H\ast G)(du)(H^{\ast n}\ast G^{\ast n}\ast T_x)(dv)\nonumber\\
&=&\varphi(q,x,m)+\int_{(0,\infty)}e^{\Phi(q)u} \,u \,\mathbb I_{\{u\geq m\}}\,H\ast G(du)\nonumber\\
	&& \sum_{n=0}^{\infty} \int_{(x,\infty)}e^{\Phi(q)v}H^{\ast n}\ast G^{\ast n}\ast T_x(dv)\nonumber\\
&=&\varphi(q,x,m)+\kappa(q,x) \;\underbrace{\int_{(m,\infty)} e^{\Phi(q)u}\, u\,H\ast G(du)}_{I_m} \nonumber\\
	&&
 \sum_{n=0}^{\infty} 
 \big[\underbrace{\int_{(0,\infty)}e^{\Phi(q)v}H \ast G(dv)}_{I}\big]^n\,,\label{id1}
\end{eqnarray}
where 
	$$\varphi(q,x,m) = \mathbb{E}\Big[ e^{-q\,\tau_x}\,(Y_{\tau_x}-x)\,\mathbb I_{\{Y_{\tau_x}-x\;\geq m\}}\,;\, 
		\tau_x < \infty  \Big]$$
and
	\begin{equation}
\kappa(q,x) = \int_{(x,\infty)} e^{\Phi(q)\,v}\, T_x(dv)
	= \mathbb{E}\big[e^{-q\,\tau_x} \,;\, \tau_x < \infty]\,,
	\label{kappa}
	\end{equation} 
since $T_x(dv) = \widetilde{\mathbb{P}}(Y_{\tau_x}\in dv,\, \tau_x < \infty)\/$, for $v>x\/$.
\\

Now the term $I$ in \eqref{id1} is equal to 
	\begin{equation}
I = \int_0^\infty e^{\Phi(q)\,u} \, H\ast G(du) = \int_0^\infty
	e^{\Phi(q)\,u} \, H(du) \,\int_0^\infty e^{\Phi(q)\,u} \, 
	G(du)\,.
	\label{egal}
	\end{equation}
The first integral in the above equation is given by 
	\begin{eqnarray}\label{egal1}
\int_{(0,\infty)} e^{\Phi(q)\,u} \,H(du)
	&=& \frac{1}{c+\Phi(q)\sigma^2} \int_0^{\infty} e^{-\Phi(q)\,y} 
		\,\nu(y,\,\infty)\,dy \nonumber\\
	&=& \frac{2c\Phi(q)+\sigma^2 \Phi(q)^2-2q}{2 \Phi(q)\big[c+\Phi(q)\sigma^2\big]} \,,
	\end{eqnarray}
where the last equality uses the identity $\psi(\Phi(q))=q\/$ and integration by parts.

The second integral in \eqref{egal} is 
	\begin{eqnarray}
\int_{(0,\infty)} e^{\Phi(q)\,u} \,G(du)
	&=& \frac{2[c+\phi(q)\sigma^2]}{2c+\phi(q)\sigma^2}\,.
	\label{egal2}
	\end{eqnarray}
By substituting \eqref{egal1} and \eqref{egal2} in \eqref{egal}, we conclude that 
	$$\int_{(0,\infty)} e^{\Phi(q)\,u} \, H\ast G(du)
		= 1-\frac{2q}{\Phi(q)\big[2c+\sigma^2  \Phi(q)\big]} \in (0,1)\,.$$

Finally the term $I_m\/$ in \eqref{id1} is equal to 

\begin{eqnarray}
I_m &=&\int_{u+v >m}e^{\Phi(q)(u+v)} (u+v)\,H(du)G(dv)\nonumber\\
&=& \underbrace{\int_{u >0}\int_{v >m}e^{\Phi(q)(u+v)} (u+v)\;H(du)G(dv)}_{I_m^1} \nonumber\\
&& \qquad +\underbrace{\int_{u >m-v}\int_{0<v <m}e^{\Phi(q)\,(u+v)} \,(u+v)\,H(du)\,G(dv)}_{I_m^2}
\nonumber\\ 
&=& I_m^1 + I_m^2\,, \label{I1+I2}
\end{eqnarray}
where
\begin{eqnarray}
I_m^1&=&\int_{(0,\infty)} e^{\Phi(q)u} \big[\int_{(m,\infty)} (u+v) \, e^{\Phi(q)v}\, G(dv)\big] \,H(du)\nonumber\\
&=&e^{-[\frac{2\,c + \phi(q)\sigma^2}{\sigma^2}]\,m}\,\int_{(0,\infty)} e^{\Phi(q)u}\, \big[u\int_0^{\infty} \frac{2\widetilde{c}}{\sigma^2} \,e^{-[\frac{2\widetilde{c}}{\sigma^2}-\Phi(q)]v}\, dv \nonumber\\
&&\qquad+\int_0^{\infty} \frac{2\widetilde{c}}{\sigma^2}\,(v+m)\,e^{-[\frac{2\widetilde{c}}{\sigma^2}-\Phi(q)]v} \,dv\big] \,H(du)
\nonumber\\
&=&e^{-[\frac{2\,c + \phi(q)\sigma^2}{\sigma^2}]m} \, \int_{(0,\infty)} e^{\Phi(q)u} \,\Big[\frac{2(c+\Phi(q)\sigma^2)}{2c+\Phi(q)\sigma^2} \,\big(u +m+\frac{\sigma^2}{2c + \Phi(q)\sigma^2}\big)\Big] \,H(du)
\nonumber\\
&=&e^{-[\frac{2\,c + \phi(q)\sigma^2}{\sigma^2}]m}\,\frac{2 }{2c+\Phi(q)\sigma^2} \,\Big[\big(\frac{\sigma^2}{2c + \Phi(q)\sigma^2}+m\big) \int_{(0,\infty)} e^{-\Phi(q)y}\,\nu(y, \infty)\,dy \nonumber \\
&&\qquad +\int_0^{\infty} \int_{(0,\infty)} u\,e^{-\Phi(q)y}\,\nu(du +y)\,dy\Big]\nonumber\\
&=&e^{-[\frac{2\widetilde{c}}{\sigma^2}+\Phi(q)]m} \,\frac{2}{2c+\Phi(q)\sigma^2} \, \Big[\big(\frac{\sigma^2}{2c + \Phi(q)\sigma^2}+m\big) \int_{(0,\infty)} e^{-\Phi(q)y}\,\nu(y, \infty)\,dy \nonumber\\
&&\qquad+\frac{1}{\Phi(q)}\,\big[\int_0^{\infty} \nu(y, \infty)\,dy - \int_0^{\infty} e^{-\Phi(q)y}\nu(y, \infty)\,dy\big]\Big]\nonumber\\
&=& e^{-[\frac{2\;c + \phi(q)\sigma^2}{\sigma^2}]m} \,\frac{2 }{2c+\Phi(q)\sigma^2}\Big[\big(\frac{\sigma^2}{2c + \Phi(q)\sigma^2}-\frac{1}{\Phi(q)}+m\big)  \nonumber \\
&& \qquad\times\int_{(0,\infty)} e^{-\Phi(q)y}\;\nu(y,\infty)dy +\frac{1}{\Phi(q)}\mathbb{E}[S_1]\Big] \nonumber\\
&=&e^{-[\frac{2\,c + \phi(q)\sigma^2}{\sigma^2}]m} \,\frac{2\;c\; \big[2q +\Phi(q)(\rho+m-1)(2c + \Phi(q)\sigma^2) \big]}{\Phi(q)^2\,[2c+\Phi(q)\sigma^2]^2}\,
\end{eqnarray}
and where in the last equality we used \eqref{egal1}. 
 
The second term $I_m^2\/$ in \eqref{I1+I2} is given by 
\begin{eqnarray}
I_m^2 &=&\int_{0<v<m} \int_{u>m-v} (u+v)\,e^{\Phi(q)u} e^{\Phi(q)v}G(dv) \,H(du)\nonumber\\
&=&\int_{0<v<m} G(dv) \int_{u>0} (u+m)\,e^{\Phi(q)(u+m)} \,H(du+m-v)\nonumber\\
&=&\frac{2}{\sigma^2} \int_{0<v<m} e^{-[\frac{2c+\Phi(q)\sigma^2}{\sigma^2}]v}\, dv \int_{u>0} \int_0^{\infty} (u+m) \,e^{-\Phi(q)y}\, \nu(du+m-v+y)\, dy \nonumber\\
&=& \frac{2}{\sigma^2} \int_{0<v<m} e^{-[\frac{2c+\Phi(q)\sigma^2}{\sigma^2}]v} dv \int_0^{\infty} \,e^{-\Phi(q)y} \int_{u>0} (u+m)\, \nu(du+m-v+y)\, dy \nonumber\\
&=&\frac{2}{\sigma^2}\int_{0<v<m} e^{-[\frac{2c+\Phi(q)\sigma^2}{\sigma^2}]v} \int_0^{\infty} \,e^{-\Phi(q)y} \big[\int_y^{\infty}  \nu(u+m-v,\,\infty)\,du \nonumber\\
&& \qquad + \nu(y+m-v,\;\infty)\big]\,dy\,dv\nonumber\\
&=&\frac{2}{\sigma^2 }\int_{0<v<m} e^{-[\frac{2c+\Phi(q)\sigma^2}{\sigma^2}]v}\, \frac{1}{\Phi(q)}\Big[\int_0^{\infty} \nu(u+m-v,\,\infty)\,du \\
&& \qquad -\int_0^{\infty} e^{-\Phi(q)y}\nu(y+m-v,\;\infty)\;dy \Big]\;dv\nonumber\\
&=&\frac{2}{\Phi(q)\sigma^2 } \int_0^m e^{-[\frac{2c+\Phi(q)\sigma^2}{\sigma^2}]v} \, \Big[ \int_0^{\infty} [1+(m\Phi(q)-1)e^{-\Phi(q)u}]\nonumber\\
&&\qquad \times \,\nu(u+m-v,\,\infty)\,du \Big]\,dv\,.
\end{eqnarray}
	
Finally, using the expressions above for $I_m\/$ in \eqref{I1+I2} and
$I\/$ in \eqref{egal}, then \eqref{id1} is equal to  
	\begin{eqnarray}\label{id12}
\Pi_1(q,x,m) &=& \varphi(q,x,m) + \underbrace{(I_m^1+I_m^2) \sum_{n=0}^\infty 
		[I]^n}_{\delta(q,\sigma,m)}\, \kappa(x,q)\nonumber \\
	&=& \varphi(q,x,m) + \delta(q,\sigma,m)\,\kappa(x,q)\,,
	\end{eqnarray} 
where
	\begin{eqnarray*}
\delta(q,\sigma,m) &=& \frac{c\,\big[2q +\Phi(q)\,(\rho+m-1)\,(2c + \Phi(q)\sigma^2) \big]}
		{q\,\Phi(q)\,[2c+\Phi(q)\sigma^2]}\,e^{-[\frac{2\;c + \phi(q)\sigma^2}{\sigma^2}]m}\\ 
	&& \qquad + \frac{[2c+\Phi(q)\sigma^2]}{q\sigma^2} \int_0^m e^{-\frac{[2c+\Phi(q)\sigma^2}
		{\sigma^2}]v} \, \Big[ \int_0^\infty [1+(m\Phi(q)-1)\,e^{-\Phi(q)u}]\\
	&& \qquad \times \, \nu(u+m-v,\,\infty)\,du \Big]\,dv\,.
	\end{eqnarray*}
Now, to complete the proof of the theorem, we need only to identify the two functions $\varphi(q,x,m)\/$ and $\kappa(x,q)\/$. Recall Lemma 4.1 in \cite{Bensalah1} that gives an explicit form for the classical EDPF, $\phi(w,q,x)\/$, defined in \eqref{edpfs1}. With it we can derive explicit expressions for the above functions $\varphi(q,x,m)\/$ and $\kappa(x,q)\/$:
	\begin{eqnarray}
\kappa(q,x) &=& \mathbb{E}\big[e^{-q\,\tau_x} \,;\, \tau_x < \infty] \nonumber\\
	&=& f\ast t(x)\,,
	\label{kappa1}
	\end{eqnarray}
where
	\begin{equation}
f(x) = W'^{(q)}(x)-\Phi(q)\,W^{(q)}(x) 
	\end{equation}
and
	\begin{equation}
t(x) = e^{\Phi(q)\,x} \int_x^\infty e^{-\Phi(q)\,v}\, \nu(v,\infty) dv\,,
	\label{t2}
	\end{equation}
while
	\begin{eqnarray}
\varphi(q,x,m) &=& \mathbb{E}\Big[ e^{-q\,\tau_x}\,(Y_{\tau_x}-x)\,\mathbb I_{\{Y_{\tau_x }-x\; \geq m \}}\,;\, \tau_x < \infty  \Big]\nonumber\\
	&=& \widetilde{\mathbb{E}}\Big[ e^{\Phi(q)\,Y_{\tau_x}}\, (Y_{\tau_x}-x) \,\mathbb I_{\{Y_{\tau_x}-x\; \geq m \}} \,;\, \tau_x < \infty  \Big] \nonumber\\
	&=& \int_0^\infty e^{\Phi(q)\,u}\,(u+m)\,T_x(du+x+m) 
		= f\ast h_m(x)\,,
	\end{eqnarray}
where $\widetilde{\mathbb{E}}\/$ is the expectation under $\widetilde{\mathbb{P}}\/$, $T_x(\cdot)\/$ is the overshoot distribution defined in \eqref{overshoot} and 
	\begin{eqnarray}
h_m(x) &=& e^{\Phi(q)\,x} \int_x^\infty e^{-\Phi(q)\,v} \int_{(0,\infty)} (u+m)\, \nu(du+v+m)\, dv\,.
		 \label{h}
	\end{eqnarray}
Then the first part of the theorem follows. \\\\
2. Following the same order of ideas above, the second part of Theorem \ref{CI} can be easily shown. In fact, the expression of the extended EDPF in \eqref{SEDPF} reduces the proportional RCI premium, since here we take $F_1(v,u)=a\,(u-x)\/$ and $F_n(v,u)=a\,(u-v)\,$, for $u \geq0\/$, $v \in \mathbb{R}\/$ and $n\geq 2\/$. Hence
	\begin{eqnarray}
\Pi_2(q,x,a) &=& \mathbb{E}\Big[e^{-q\,\tau_x}\,a\,(Y_{\tau_x}-x)\,;\,\tau_x < \infty\Big]
		+ \nonumber \\
	&& \qquad \sum_{n=0}^\infty \int_x^\infty \int_0^\infty e^{\Phi(q)\,(u+v)} \,a \, 
		u\, H\ast G(du) \, H^{\ast n}\ast 
		G^{\ast n}\ast T_x(dv)\nonumber\\
	&=& a\,\varphi(q,x,0) + a\, \kappa(q,x) \, \underbrace{\int_0^\infty e^{\Phi(q)\,(u+v)}
		\,u\, H\ast G(du)}_{I_0} \nonumber\\
	&&\qquad \times \,\sum_{n=0}^\infty \big[\underbrace{\int_0^\infty e^{\Phi(q)\,v}
		H \ast G(dv)}_{I}\big]^n \nonumber\\
	&=& a\,\Pi_1(q,x,0) = a\, \varphi(q,x,0) + a\,I_0 \, \sum_{n=0}^\infty [II]^n\,
		\kappa(q,x)\nonumber\\
	&=& a\, \varphi(q,x,0) + a\, \underbrace{\frac{c\big[2q+\phi(q)(\rho-1)
		[2c+\phi(q)\sigma^2]\big]}{q\phi(q)[2c+\phi(q)\sigma^2]}}_{\delta
		(q,\sigma,0)}\, \kappa(q,x)\nonumber \\
	&=& a\, \varphi(q,x,0) + a\, \delta(q,\sigma,0)\,\kappa(q,x) \,. \label{id2}
	\end{eqnarray}
	\end{proof}\\

The following section illustrates the results above for two particular cases, including the Cram\'er--Lundberg risk model, without a Brownian component.

\section{Examples: two classical risk models}\label{classical}

We study in this section particular examples of risk process $Y\/$ satisfying the general setting described in Section \ref{riskEDPF} for which the $q$--scale function has a tractable form.  These provide some interesting examples of insurance models with relatively simple expressions for the reinsurance by capital injections (RCI) premiums. 

In fact, a tractable form for the $q$--scale function is here inherited by the functions $\varphi(q,x,m)\/$ and $\kappa(q,x)\/$, defined in \eqref{varphi} and \eqref{kappa} respectively. These functions are key ingredients in the general expressions of Theorem \ref{CI}. In what follows we analyze in more detail some models for which we can have an explicit understanding of the reinsurance premium problem:
	\begin{itemize}
\item the classical Cram\'er--Lundberg model with exponential claims,
\item the spectrally negative stable risk process.
	\end{itemize} 


\subsection{Classical Cram\'er--Lundberg model with exponential claims}\label{classic_case}


The so--called classical or Cram\'er--Lundberg model was introduced in \cite{lundberg}. The surplus process is a compound Poisson process starting at $x\geqslant 0$, i.e.,
	\begin{equation}
R_t = x + ct - \sum_{i=1}^{N_t} Z_i\,,
	\label{classrisk}
	\end{equation}
where the number of claims is assumed to follow a Poisson process $(N_t)_{t\geqslant 0}\/$ with intensity $\lambda\/$, independent of the positive and \emph{iid} random variables $(Z_n)_{n\geqslant 1}\/$ representing claim sizes. The loaded premium $c\/$ is of the form $c= (1+\theta)\,\lambda\,\mathbb{E}[Z_1]\/$ for some safety loading factor $\theta>0\/$. The form of the $q$--scale function in this model is relatively simple when claim sizes are exponentially distributed with mean $1/\mu\/$. In this case, the L\'evy measure takes the simple form $\nu(dx)=\lambda\,\mu\, e^{-\mu x}dx\/$. In turn, the Laplace exponent in (\ref{khinchine}) becomes
	\begin{equation}
\psi(s) = c\,s - \frac{\lambda\,s}{\mu+s}\,, \qquad s>0\,.
	\label{laplace_classical}
	\end{equation}
So, the premium rate is $c=\lambda\,(1+\theta)/\mu\/$ where $\theta>0\/$ is a positive security loading.

This model has been for long a textbook example for which the distribution of ruin--related quantities can be explicitly computed. Here, we study in detail the RCI reinsurance premium for this particular example. Moreover, we derive explicit expressions for the RCI premiums from Theorem \ref{CI}.

The expression for the $q$--scale function in this case is known (see \cite{Kuz_Kyp_Riv}) and is given by 
	\begin{equation}
W^{(q)}(x) = \frac{\mu+\Phi(q)}{\eta_q } \, e^{\Phi(q)\,x} - \frac{\mu+\Theta(q)}{\eta_q }\,
	e^{\Theta(q)\,x}\,,
	\label{scalefunction}
	\end{equation}
where $\Phi(q)\/$ and $\Theta(q)\/$ are the solutions of $\psi(s)=q\/$, i.e.~$\Phi(q) = \frac{1}{2c}\,(q+\lambda -c\,\mu +\eta_q\/)$ and $\Theta(q)= \frac{1}{2c}\,(q + \lambda -c\,\mu -\eta_q\/)$, with $\eta_q = \sqrt{(q+\lambda-c\,\mu)^2+4q\,\mu\,c}\/$. For details see \cite{Kuz_Kyp_Riv}.

In this case, the expressions for $\kappa(q,x)\/$ and $\varphi(q,x,m)\/$ in \eqref{kappa} and \eqref{varphi} are  given by
	\begin{equation}
\varphi(q,x,m) = f\ast h_m(x)
	\label{overshoot2}
	\end{equation}
and
	\begin{equation}
\kappa(q,x) = f\ast t(x)\,,
	\label{k2}
	\end{equation}
where
	\begin{eqnarray}
f(x) &=& \frac{[\mu+\Theta(q)]\,[\Phi(q)-\Theta(q)]}{\eta_q}\, e^{\Theta(q)\,x}\nonumber\\
	&=& \Big[\frac{\mu+\Theta(q)}{c}\Big]\, e^{\Theta(q)x}\,,
	\label{k2}	
	\end{eqnarray}
	\begin{equation}
t(x) = \frac{\lambda}{(\Phi(q)+\mu)}\,e^{-\mu\, x}
	\label{t2}
	\end{equation}
and
	\begin{equation}
h_m(x) = \frac{\lambda\,[m\mu +1]}{\mu\,[\Phi(q)+\mu]}\,e^{-\mu\, (x+m)}\,.
	\label{h2}
	\end{equation}
Hence the expressions of $\kappa(q,x)\/$ and $\varphi(q,x,m)\/$ reduce to 
	\begin{equation}
\varphi(q,x,m) = \frac{\lambda\,e^{-\mu\,m}\,[m\mu +1]}{\mu\,c\,[\mu+\Phi(q)]}[ e^{\Theta(q)\,x} - e^{-\mu\,x}]
	\label{overshoot3}
	\end{equation}
and 
	\begin{equation}
\kappa(q,x) = \frac{ \lambda}{c\,[\mu+\Phi(q)]}[ e^{\Theta(q)\,x} - e^{-\mu\,x}]\,.
	\label{k3}
	\end{equation}
Recall that in this model, the Brownian component vanishes and then the distribution $G(\cdot)\/$ reduces to a Dirac measure at $0\/$. Hence, the expression for $I_m\/$ in \eqref{I1+I2} is equal to 
	\begin{equation}
\int_{u >m} e^{\Phi(q)\,u}\,u \,H(du) = \frac{\lambda\,(m\mu +1)}{c\,\mu\,[\Phi(q)+\mu]}\,e^{-\mu m}\,,
	\label{k3}
	\end{equation}
and $II$ reduces to
	\begin{eqnarray*}
\int_{u >0} e^{\Phi(q)\,u} \,H(du) &=& 1-\frac{q}{\Phi(q)\,c}\,.
	\end{eqnarray*}

Using Theorem \ref{CI}, we provide explicit expressions for the RCI reinsurance premiums. In fact, in the case of the classical model $\sigma=0\/$ and hence the extreme--loss RCI premium in \eqref{V2} reduces to
	\begin{eqnarray}
\Pi_{1}(q,x,m) &=& \varphi(q,x,m) + \frac{\lambda \Phi(q)\,(m\mu +1)\, e^{-\mu m}}{\mu\,q\,[\mu+\Phi(q)]}\, \kappa(q,x)\nonumber\\
&=& \frac{\lambda e^{-\mu m}\,(m\mu +1)}{\mu\,c\,[\mu+\Phi(q)]} \, \big[e^{\Theta(q)x} - e^{-\mu x}\big] \,\Big(1+\frac{\lambda \Phi(q)}{q[\mu+\Phi(q)]} \Big)\nonumber\\
&=& \frac{\lambda\,\Phi(q) e^{-\mu m}\,(m\mu +1)}{\mu\,q\,[\mu+\Phi(q)]}\,\big[e^{\Theta(q)x} - e^{-\mu x}\big]\,.
\label{V0}
\end{eqnarray}
Similarly, the proportional RCI premium in \eqref{V} reduces to
	\begin{eqnarray}
\Pi_{2}(q,x,a)&=&a\Pi_{1}(q,x,0)\nonumber\\
&=&\frac{\lambda \,a\,\Phi(q)}{\mu\,q\,[\mu+\Phi(q)]}\, \big[e^{\Theta(q)x} - e^{-\mu x}\big]\,,
	\label{V1}
	\end{eqnarray}
where, recall, $q\/$ is the discount factor, $x\/$ the initial surplus and $a <1\/$ the factor applied to the risks ceded when a reinsurance capital injection is needed.
 
\subsection{Spectrally negative stable process}
\label{sec: Stable}

In this subsection, we study the case when the surplus process is driven by a spectrally negative stable process with stability parameter $\alpha \in (1,2)\/$. This model was studied in the insurance context in \cite{Furrer}. We calculate here the RCI reinsurance premium and give explicit expressions for both, proportional and extreme--loss RCI contracts.

Let $(X_t)_{t\geq 0}\/$ be a spectrally negative stable process with stability parameter $\alpha \in (1,2)\/$ and Laplace exponent $\psi(s) =(s+c)^{\alpha}-c^\alpha\/$, for $c,\,s>0\/$. Here the L\'evy measure in \eqref{khinchine} is given by $\nu(dx)= \frac{e^{-c\,x}}{x^{1+\alpha}\, \Gamma(-\alpha)}\,dx\/$, for $x>0\/$ and $\Gamma(u)\/$ is the \emph{gamma} function. It can be seen (for example \cite{Bensalah2, Biffis}) that  
	\begin{equation}
W^{(q)}(x) = e^{-c\,x}\, x^{\alpha -1}\, E_{\alpha, \alpha}\big[(q+c^{\alpha})\, x^\alpha\big]\,,
	\label{stable1}
	\end{equation}
for $x,\,q\geqslant 0\/$, where $E_{\alpha, \beta }(z) = \sum_{k\geqslant 0} \frac{z^k}{\Gamma(\beta+\alpha k)}\/$ is the two--parameter Mittag--Leffler function. It is clear that here $\Phi(q)=(q+c^{\alpha})^{\frac{1}{\alpha}}-c\/$.

The expressions for $\kappa(q,x)\/$ and $\varphi(q,x,m)\/$ then reduce to 
	\begin{eqnarray}
&& \varphi(q,x,m) = \frac{e^{-c\,m}}{x^{1+\alpha}\,\Gamma(-\alpha)} \int_0^\infty \int_0^\infty 
		(u+m)\,e^{[(q+c^{\alpha})^{1/\alpha}-c]\,(x-v)}\, \frac{e^{-c\,(u+v)}}
		{(u+v+m)^{\alpha}} \nonumber\\
&& \qquad \times \, \Big[e^{-(q+c^{\alpha})^{1/\alpha}\,x} \,x^{\alpha-1} 
		\,E_{\alpha, \alpha}\big[(q+c^{\alpha})\, x^{\alpha}\big]\Big]\, du\,dv\nonumber\\
&& \qquad - \frac{e^{-c\,m}}{x^{1+\alpha}\, \Gamma(-\alpha)} \int_0^\infty \int_0^x
		(u+m)\, e^{[(q+c^{\alpha})^{1/\alpha}-c]\,(x-v)}\, \frac{e^{-c\,(u+v)}}
		{(u+v+m)^{\alpha}} \nonumber\\ 
&&\qquad \times \, \Big[e^{-(q+c^{\alpha})^{1/\alpha}\,(x-v)}\,(x-v)^{\alpha-1}\,
		E_{\alpha, \alpha}\big[(q+c^{\alpha})\, (x-v)^{\alpha}\big]\Big]\,du\,dv\,, 
		\label{a2}
	\end{eqnarray}
and 
	\begin{eqnarray}
&&\kappa(q,x) = \frac{1}{x^{1+\alpha}\,\Gamma(-\alpha)} \int_0^\infty \int_0^\infty 
		e^{[(q+c^{\alpha})^{1/\alpha}-c]\,(x-v)}\, \frac{e^{-c\,(u+v)}}{(u+v)^{\alpha}}
		\nonumber\\ 
&&\quad \times \, \Big[e^{-(q+c^{\alpha})^{1/\alpha}\,x}\,x^{\alpha-1} 
		\,E_{\alpha, \alpha}\big[(q+c^{\alpha})\, x^{\alpha}\big]\Big]\,du\,dv\nonumber\\
&&\quad - \frac{1}{x^{1+\alpha}\, \Gamma (-\alpha)}\int_0^\infty \int_0^x 
		e^{[(q+c^{\alpha})^{1/\alpha}-c]\,(x-v)} \,\frac{e^{-c\,(u+v)}}{(u+v)^{\alpha}}
		\nonumber\\ 
&&\quad \times \Big[e^{-(q+c^{\alpha})^{1/\alpha}\,(x-v)}\,(x-v)^{\alpha-1}\,E_{\alpha, \alpha}
		\big[(q+c^{\alpha})\,(x-v)^{\alpha}\big]\Big]\,du\,dv\,.
	\label{k2}
	\end{eqnarray}
As in the previous model, the distribution $G(\cdot)\/$ reduces to a Dirac measure at $0\/$. Hence, the expression for $I_m\/$ in \eqref{I1+I2} here is equal to 
	$$I_m = \int_{u >m} e^{\Phi(q)\,u}\,u\,H(du) 
		= \frac{e^{-c\,m}}{\Gamma(-\alpha)} \,\int_0^\infty \int_0^\infty 
		\frac{e^{-c\,(u+m)}\,e^{-(q+c^{\alpha})^{1/\alpha}\,y}}{(u+y+m)^{1+\alpha}}\,dy\,du\,,$$
and $I\/$ in \eqref{egal} reduces to here
	$$I = \int_{u >0} e^{\Phi(q)\,u} \,H(du) = 1-\frac{q}{\Phi(q)\,c}
		= 1-\frac{q}{c\,[(q+c^{\alpha})^{1/\alpha}-c]}\,.$$
Using Theorem \ref{CI}, we can give explicit expressions for the RCI reinsurance premiums.  The extreme--loss RCI premium in \eqref{V2} reduces to
	\begin{eqnarray}
\Pi_1(q,x,m) &=& \varphi(q,x,m) + \kappa(q,x)\,\Big[\frac{c\,e^{-c\,m}[(q+c^{\alpha})^{1/\alpha}-c]}
		{q\,\Gamma(-\alpha)} \Big] \nonumber\\
	&&  \times\, \Big[\int_0^\infty \int_0^\infty \frac{(u+m)\,e^{-c\,u}\,
		e^{-(q+c^{\alpha})^{1/\alpha}\,y}}{(u+y+m)^{1+\alpha}}\,dy\,du\Big]\,.
	\end{eqnarray}
Similarly, the proportional RCI premium in \eqref{V} reduces to
	\begin{eqnarray}
\Pi_2(q,x,a) &=& a\,\Pi_1(q,x,0)\nonumber\\
	&=& a\,\varphi(q,x,0) + \Big[\frac{a\, c[(q+c^{\alpha})^{1/\alpha}-c]}{q\,\Gamma(-\alpha)}\Big]
		\nonumber\\
	&&\times\, \Big[\int_0^\infty \int_0^\infty \frac{u\,e^{-c\,u}\,
		e^{-(q+c^{\alpha})^{1/\alpha})\,y}}{(u+y)^{1+\alpha}}\,dy\,du\Big]\,\kappa(q,x)
		\label{V1}\,.
\end{eqnarray}
where $\varphi(q,x,m)\/$ and $\kappa(q,x)\/$ are given respectively by \eqref{a2} and \eqref{k2}, which are sufficiently explicit to evaluate numerically with such programs as \emph{Maple\/}, \emph{Matlab\/} or \emph{Mathematica\/}. 

\subsection{Numerical examples}
\label{sec: examples}

To show the tractability of the RCI premium formulas derived above, in the two previous sections, we include here a few numerical illustrations. The premium values were obtained in Python; the code is available to interested readers upon request.

For these numerical illustrations the Poisson parameter of the classical risk model in \eqref{classrisk} was set to $\lambda=1\/$ as well as the mean exponential claim size $\mu=1\/$. Figures \ref{fig:subfigures}--\ref{fig:subfigures2} show the effect on the extreme--loss RCI premiums $\Pi_1(q,x,m)\/$, in \eqref{V0}, of varying the remaining decision variables, such as the RCI retention level $m\/$, the discount factor $q\/$, the premium safety loading $\theta\/$ and the initial surplus $x\/$. 

\begin{figure}[ht!]
     \begin{center}
        \subfigure[$\Pi_1(q,x,m)\/$ curves by discount factor $q\/$]{%
            \label{fig:first}
            \includegraphics[width=0.5\textwidth]{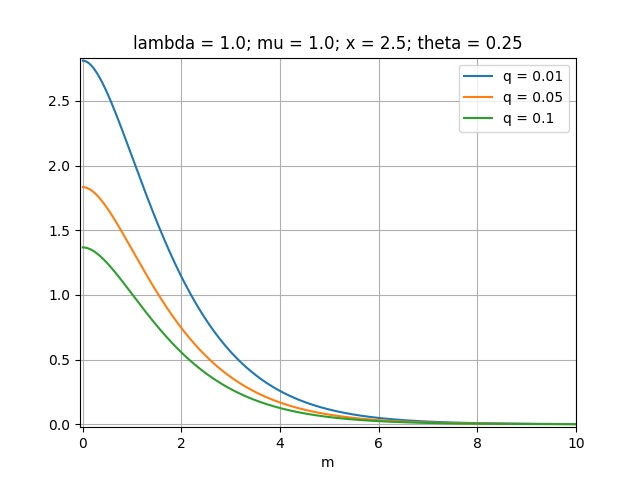}
        }%
        \subfigure[$\Pi_1(q,x,m)\/$ curves by loading $\theta\/$]{%
           \label{fig:second}
           \includegraphics[width=0.5\textwidth]{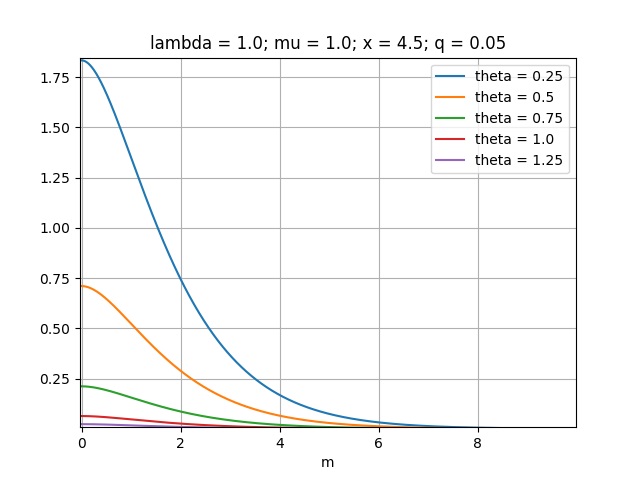}
        }\\ \vspace{-0.1in} 
        \subfigure[$\Pi_1(q,x,m)\/$ curves by initial surplus $x\/$]{%
            \label{fig:third}
            \includegraphics[width=0.5\textwidth]{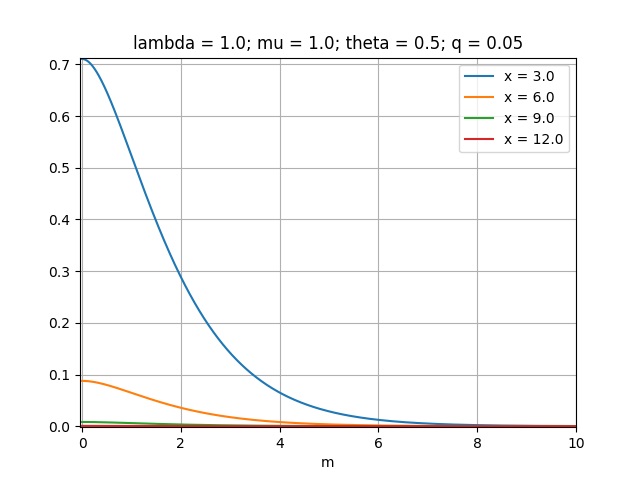}
        }%
    \end{center}
\vspace{-0.2in}
\caption{Extreme--loss RCI premiums $\Pi_1(q,x,m)\/$ versus retention level $m\/$}    
   \label{fig:subfigures}
\end{figure}
\noindent
First, Figure \ref{fig:first} plots curves of $\Pi_1(q,x,m)\/$ versus varying values of the RCI retention level $m\/$, for different discount factors $q\/$, with an initial surplus of $x=2.5\/$ and safety loading factor of $\theta=0.25\/$. Then Figure \ref{fig:second} studies the effect of the loading factor $\theta\/$ on $\Pi_1(q,x,m)\/$ for an initial surplus of $x=4.5\/$ and discount factor $q=0.05\/$, again versus  varying values of the RCI retention level $m\/$. Similarly, Figure \ref{fig:third} shows curves against the retention level $m\/$ for different initial surplus values $x\/$, with a loading of $\theta=0.5\/$ and discount factor $q=0.05\/$.

Figure \ref{fig:fourth} gives RCI premiums as a function of the security loading $\theta\/$, for different values of the retention levels $m\/$, with initial surplus $x=4.0\,c\/$ and discount factor $q=0.05\/$, while Figure \ref{fig:fifth} gives RCI premium curves, also as a function of $\theta\/$, for different $q\/$ values and retention level $m=1\/$. Finally Figure \ref{fig:sixth} gives RCI premium curves, as a function of $\theta\/$, for different initial surplus values $x\/$.

\begin{figure}[ht!]
     \begin{center}
        \subfigure[$\Pi_1(q,x,m)\/$ curves by retention level $m\/$]{%
            \label{fig:fourth}
            \includegraphics[width=0.5\textwidth]{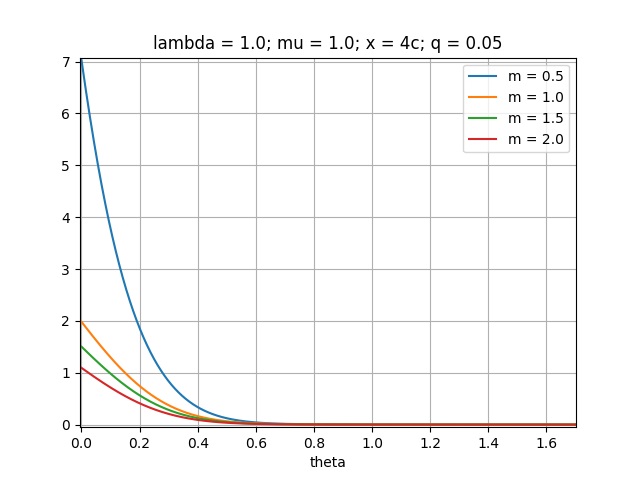}
        }%
        \subfigure[$\Pi_1(q,x,m)\/$ curves by discounting factor $q\/$]{%
            \label{fig:fifth}
            \includegraphics[width=0.5\textwidth]{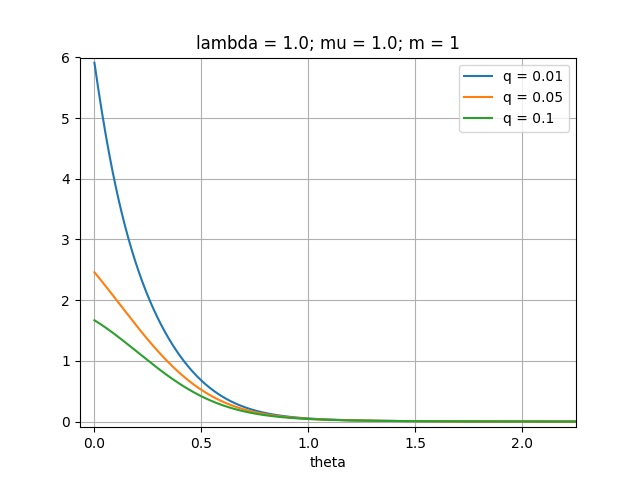}
        }\\ \vspace{-0.1in} 
        \subfigure[$\Pi_1(q,x,m)\/$ curves by initial surplus $x\/$]{%
            \label{fig:sixth}
            \includegraphics[width=0.5\textwidth]{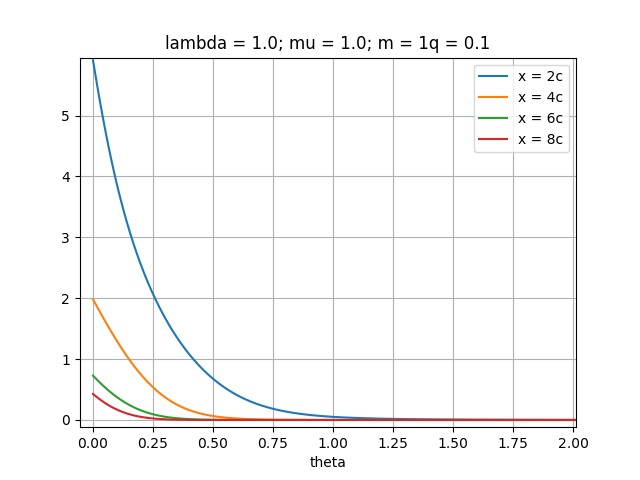}
        }%
    \end{center}
\vspace{-0.2in}
\caption{Extreme--loss RCI premiums $\Pi_1(q,x,m)\/$ versus loading factor $\theta\/$}    
   \label{fig:subfigures2}
\end{figure}
\noindent
From the curves of $\Pi_1(q,x,m)\/$ in Figures  \ref{fig:subfigures} and  \ref{fig:subfigures2} we see that:
	\begin{itemize}
\item Doubling the safety loading $\theta\/$ in Figure \ref{fig:first} or the initial surplus $x\/$ in Figure \ref{fig:second} has a greater effect on RCI premiums than increasing the discounting rate $q\/$ by a factor of 10 in Figure \ref{fig:third}.

\item The same is seen from the curves of $\Pi_1(q,x,m)\/$ versus the safety loading $\theta\/$ in Figure \ref{fig:subfigures2}, the impact of the discounting factor $q\/$ is smaller than that of varying the retention level $m\/$ or that of the initial surplus $x\/$. 

\item However, an analysis without discounting ($q=0\/$) is not possible for small values of the RCI retention level $m\/$ or safety loading $\theta\/$; the present value increases without bound with the more frequent ruin events.
	\end{itemize}
    
\section*{Conclusion}

We consider here a new type of reinsurance contract (RCI) that provides capital injections only in extreme, worse scenario cases, based on the insurer's financial position. Ruin serves as a simplifying proxy for the insurer's financial health. Reinsurance capital injections made after each ruin event allow the insurance company to continue operate indefinitely, as in an on-going concern basis over an infinite horizon.

Borrowing from recent developments in the actuarial and financial literature on models for capital injections (e.g.~Einsenberg and Schmidli, 2011) and the formulas for the expected present value of future capital injections in a quite general class of risk models (Ben Salah, 2014) we develop fair lump sum net premiums for two types of RCI contracts. In this first study we show that tractable formulas can be derived for RCI premiums so that both, insurance and reinsurance companies, can compare the cost of RCI contracts to their alternative risk mitigation strategies/products. We also show that an analysis with discounting leads to unstable numerical calculations when ruin events become more frequent, as for example when small RCI retention levels or small safety loadings $\theta\/$ are chosen. 

Further research should tackle the practical problems of premium allocation to finite (e.g.~one year) contract terms, defining reserves and designing other types of RCI contracts, apart from the proportional and extreme--loss agreements studied here. The problem of optimal control of these new RCI contracts is also a natural question to be studied if these treaties turn out to be useful and viable; in particular 
the optimal stochastic control strategies to choose the appropriate reinsurance retention 
levels and the capital raised from stockholders to temporarily cover the Brownian 
oscillations. The latter should be easier to control as these will occur continuously around 
the ruin barrier.

\section*{Acknowledgments}

The authors are sincerely grateful to Prof. Christian Hipp for his suggestions on an earlier version of this paper and to the two anonymous reviewers for their constructive comments. This research was partially completed during a sabbatical visit of the second author to the University Carlos III of Madrid, in Spain.

\end{document}